%% file: paper.tex
\documentclass[letterpaper, 10 pt, conference]{ieeeconf}  

\IEEEoverridecommandlockouts                              
\input{Usepackages.tex}

\usepackage{amsthm}

\tikzset{>={Stealth[]}} 

\input{mycommands.tex}

\newcommand{\zbar}{\bar z}
\makeatletter
\newcommand{\changes}[1]{%
  \ifmmode
    \textcolor{black}{#1}
  \else
    \textcolor{black}{\textnormal{#1}}
  \fi
}
\makeatother

\theoremstyle{plain}
\newtheorem{theorem}{Theorem}

\theoremstyle{remark}
\newtheorem{remark}{Remark}

\RequirePackage[
        backend=biber,
        style=ieee,
        firstinits=true,
        uniquename=init,
        maxbibnames=99,
        date=year,
        url=false,
        isbn=false,
        doi=false
        ]{biblatex}
\AtBeginBibliography{\footnotesize}

\title{\LARGE \bf
Backstepping Design of Dynamic State Feedback Controllers for Parabolic Systems
}

\graphicspath{{../Figures/}}

\author{Nicole Gehring \and Benedikt Schwämmle* \and Abdurrahman Irscheid \and Joachim Deutscher
\thanks{* corresponding author}
\thanks{N. Gehring and B. Schwämmle are with the Chair of Systems Theory and Control Engineering,
        Otto von Guericke University Magdeburg, Germany
        {\tt\small \{nicole.gehring, benedikt.schwaemmle\}@ovgu.de}}%
\thanks{A. Irscheid is with the Department of Mechanical Engineering,
        American University of Sharjah, United Arab Emirates
        {\tt\small airscheid@aus.edu}}%
\thanks{J. Deutscher is with the Institute of Measurement, Control and Microtechnology,
        Ulm University, Germany
        {\tt\small joachim.deutscher@uni-ulm.de}}%
}

\begin{document}

\maketitle
\thispagestyle{empty}
\pagestyle{empty}

\begin{abstract}
Recently, dynamic state feedback controllers that are based on dynamic extensions have been presented for heterodirectional hyperbolic systems.
In this paper, a similar concept for the control of coupled diffusion-reaction systems is suggested.
The introduction of a specific controller dynamics leads to homogenized diffusion coefficients for the extended system.
Then, a backstepping-based static state feedback for the dynamically extended system is designed, which, overall, results in a dynamic state feedback.
Such a design allows stabilizing a more general class of parabolic systems as well as assigning arbitrary closed-loop dynamics.
This can be used, e.g., to achieve a decoupled input-output behavior, which is, in general, not possible with a static state feedback.
A simulation example illustrates the results.

\end{abstract}

\begin{keywords}

   Backstepping, parabolic systems, dynamic controllers, distributed-parameter systems

\end{keywords}

\section{INTRODUCTION}
\label{sec:intro}
Since its introduction, backstepping has become a powerful tool for controlling both hyperbolic and parabolic infinite-dimensional systems, see \cite{krsticBoundaryControlPDEs2008} as well as \cite{vazquezBacksteppingPartialDifferential2026} for a general overview.
While the concept was initially only presented for a scalar parabolic system, \cite{baccoliBoundaryControlCoupled2014} and \cite{baccoliBoundaryControlCoupled2015} extended the results to systems consisting of $n$ coupled diffusion-reaction equations with identical and mutually different diffusion coefficients, respectively.
The case of multivariable parabolic systems with mutually different \emph{spatially varying} diffusion coefficients was presented in \cite{vazquezBoundaryControlCoupled2017} for Dirichlet boundary conditions and in \cite{deutscherBacksteppingControlCoupled2018} for Robin boundary conditions.

In all the aforementioned cases, static state feedback was used as the stabilizing control law, which means that the controller only makes use of the state at the current time.
The main limitation of backstepping-based static state feedback lies in the fact that systems with mutually different diffusion coefficients require specific terms in the target system for the well-posedness of the kernel equations.
A more general concept offering additional design opportunities is dynamic state feedback, where the controller contains its own dynamics.
For finite-dimensional (nonlinear) systems, dynamic controllers are already well known and can be used for many purposes \cite{isidoriNonlinearControlSystems1995}.
Among others, they are capable of output tracking and disturbance rejection by using an internal model.
Another application is input-output decoupling of multi-input multi-output systems, which is in general not feasible with a static state feedback.
\changes{Observer-based compensators are also dynamic feedback controllers.
In contrast to dynamic \emph{state} feedback controllers considered in this paper, observer-based compensators use the system \emph{output} as input to the dynamic controller.
Therefore, they are of no further interest here.}

For infinite-dimensional systems, very few dynamic state feedback controllers can be found in the literature and most of them are finite-dimensional, see, e.g., \cite{ramirezExponentialStabilizationBoundary2014} for a port-Hamiltonian approach.
Regarding infinite-dimensional dynamic controllers, a breakthrough was achieved in \cite{redaudIndomainDampingAssignment2022} and \cite{redaudDomainDissipationAssignment2024}, where the port-Hamiltonian framework is combined with backstepping for hyperbolic systems.
By including delays in the control law, an additional controller dynamics is introduced implicitly.
This way, the choice of admissible target systems of the backstepping transformation is less restricted compared to when using a static state feedback.
Based on these results, controllers that build on dynamic extensions and backstepping were introduced in \cite{gehringUsingDynamicExtensions2025} for heterodirectional hyperbolic systems.
The idea of these dynamic controllers is to homogenize the transport delays by introducing a dynamic extension, which corresponds to the controller dynamics.

To the best of our knowledge, there do not exist infinite-dimensional dynamic state feedback controllers for general linear diffusion-reaction systems so far.
While it is not obvious how the design in \cite{redaudIndomainDampingAssignment2022} and \cite{redaudDomainDissipationAssignment2024} can be transferred to parabolic partial differential equations (PDEs), the main ideas from \cite{gehringUsingDynamicExtensions2025} can systematically be extended to diffusion-reaction systems, which is done here.
This not only allows stabilizing systems that are coupled at the unactuated boundary, which is not possible with existing backstepping-based controllers, but offers significantly more freedom in the choice of the closed-loop dynamics.
In particular, no additional terms have to be introduced in the target system, i.e., the closed-loop dynamics, to ensure well-posedness of the kernel equations.

\vspace{1ex}
\emph{Notation:}
For ease of notation, we write partial derivatives of a function $f$ w.r.t. its arguments as $\partial_z f(z,\zeta)$ and $\partial_\zeta f(z,\zeta)$.
The total derivative w.r.t. $z$ is denoted as $\mathrm{d}_z$.
If a function depends on only one variable, we use the notation $f'(z) = \dd_z f(z)$ for its derivative.

\section{SYSTEM DESCRIPTION AND OBJECTIVE}
\label{sec:problem_statement}
First, we introduce the system under consideration and we describe the backstepping-based design of static state feedback based on \cite{deutscherBacksteppingControlCoupled2018} as well as its limitations, which motivate the use of dynamic state feedback.

\subsection{System Definition}
\label{subsec:sys_orig_def}
We consider the general diffusion-reaction system
\begin{subequations}
\label{eq:sys_orig}
\begin{align}
        \partial_t x(z,t) &= \Lambda(z) \partial_z^2 x(z,t) + A(z) x(z,t) \label{eq:sys_orig_PDE} \\
        \partial_z x(0,t) &= Q_0 x(0,t) \label{eq:sys_orig_BC0} \\
        \partial_z x(1,t) &= Q_1 x(1,t) + u(t) \label{eq:sys_orig_BC1}
\end{align}
\end{subequations}
defined on the unit length interval with $z \in [0,1]$ and $t \geq 0$, where $x(z,t) \in \Rset^n$ is the state with the initial condition $x(z,0) = x_0(z)$ and $u(t) \in \Rset^n$ is the control input.
The matrix $\Lambda = \diag\left(\lambda_1,\dots,\lambda_n\right) \in \left(C^2([0,1])\right)^{n \times n}$ comprises the diffusion coefficients $\lambda_i$ that are sorted in descending order, i.e., $\lambda_1(z) > \dots > \lambda_n(z) > 0 , \, \forall z \in [0,1]$.
The reaction matrix $A \in \left(C^1([0,1])\right)^{n \times n}$ leads to an in-domain coupling of the diffusion processes.
We consider Robin boundary conditions (BCs) with boundary coupling matrices $Q_0$, $Q_1 \in \Rset^{n \times n}$, although the design in Section~\ref{sec:dyn_cont} extends to Dirichlet or Neumann BCs in an analog way.

\subsection{Static State Feedback}
\label{subsec:static_controller}
For the design of a static state feedback, we follow \cite{deutscherBacksteppingControlCoupled2018}, where $Q_0$ in \eqref{eq:sys_orig_BC0} is restricted to being a diagonal matrix.
To stabilize system \eqref{eq:sys_orig}, it is transformed into the target system
\begin{subequations}
        \label{eq:sys_target_static}
        \begin{align}
                \partial_t \tilde x(z,t) &= \Lambda(z) \partial_z^2 \tilde x(z,t) - \mu \tilde x(z,t) + A_0(z) \tilde x(0,t) \label{eq:sys_target_static_PDE} \\
                \partial_z \tilde x(0,t) &= 0 \label{eq:sys_target_static_BC0} \\
                \partial_z \tilde x(1,t) &= \bar v(t) \label{eq:sys_target_static_BC1}
        \end{align}
\end{subequations}
using the Volterra integral transformation
\begin{align}
        \tilde x(z,t) &= x(z,t) - \int_0^z K(z,\zeta) x(\zeta,t) \dint{\zeta}, \label{eq:Volterra_static}
\end{align}
with the kernel $K(z,\zeta) \in \Rset^{n \times n}$ defined on $0 \leq \zeta \leq z \leq 1$, and the static state feedback
\begin{align}
        \begin{split}
        u(t) &= \bar v(t) -(Q_1 - K(1,1)) x(1,t) \\ \qquad &+ \int_0^1 \partial_z K(1,\zeta) x(\zeta,t) \dint{\zeta}. \label{eq:static_control_law}
        \end{split}
\end{align}
In order to render system \eqref{eq:sys_target_static} exponentially stable in the (weighted) $L_2$ norm, the design parameter $\mu > 0$ is used to prescribe the decay rate and the input $\bar v(t)$ is set to zero. 

It is shown in \cite{deutscherBacksteppingControlCoupled2018} that the kernel equations for $K(z,\zeta)$ are well-posed and admit a unique piecewise $C^2$-solution if the matrix $A_0(z) \in \Rset^{n \times n}$ in \eqref{eq:sys_target_static_PDE} is a strictly lower triangular matrix, $Q_0$ is diagonal and additional BCs for well-posedness of the kernel equations are imposed.
The nonzero entries of $A_0(z)$ are defined by the kernel.
Due to its structure, this matrix leads to a coupling between the different equations in \eqref{eq:sys_target_static}.
In the equi-diffusivity case $\lambda_1(z) = \dots = \lambda_n(z) , \, \forall z \in [0,1]$, the matrix $A_0(z)$ is not necessary to ensure well-posedness \cite{baccoliBoundaryControlCoupled2014}, i.e., $A_0(z) = 0$.

\subsection{Dynamic State Feedback}
\label{subsec:Motivation}
Since the nonzero entries of $A_0(z)$ are defined by the kernel, we cannot assign an arbitrary target system \eqref{eq:sys_target_static}, which restricts the dynamics of how the state converges in the closed loop.
To overcome this limitation that is caused by the mutually different diffusion coefficients, we use a dynamic state feedback.
The main idea of this dynamic controller design is to ho\-mo\-ge\-nize the diffusion coefficients by transforming the spatial domain such that all diffusion coefficients are identical to the smallest coefficient $\lambda_n$, resulting in equi-diffusivity of all diffusion processes.
The design of a dynamic state feedback, with a controller dynamics that follows from a dynamic extension of \eqref{eq:sys_orig}, involves four steps:
\begin{enumerate}
        \item A preliminary transformation eliminates the in-domain coupling in \eqref{eq:sys_orig_PDE} to simplify the introduction of the dynamic extension.
        \item The diffusion coefficients are homogenized by scaling the spatial domains of each state component individually so that all diffusion processes share the same diffusion coefficient $\lambda_n$.
        \item The scaled domains are extended to the full domain $[0,1]$ by introducing a dynamic extension.
        \item A static state feedback for the dynamically extended system is designed using backstepping, which results in a dynamic state feedback for the original system.
\end{enumerate}
This design does not require any restriction of $Q_0$, e.g., being diagonal like in \cite{deutscherBacksteppingControlCoupled2018}.
Therefore, we can stabilize a larger class of diffusion-reaction systems, especially systems that are coupled at the unactuated boundary.
Further, it allows assigning a general closed-loop dynamics, which can be used to achieve decoupled input-output transfer behavior.

\section{DESIGN OF DYNAMIC CONTROLLERS}
\label{sec:dyn_cont}
The following design of a dynamic state feedback is split into the four steps listed in Section~\ref{subsec:Motivation}.

\subsection{Preliminary Transformation}
\label{subsec:decoup_trafo}
To simplify the introduction of the dynamic extension, the in-domain coupling due to the reaction term $A(z) x(z,t)$ in \eqref{eq:sys_orig_PDE} is eliminated.
For this, using the Volterra integral transformation
\begin{align}
        \tilde x(z,t) &= x(z,t) - \int_0^z K(z,\zeta) x(\zeta,t) \dint{\zeta}, \label{eq:Volterra_dynamic_prelim}
\end{align}
with the kernel $K(z,\zeta) \in \Rset^{n \times n}$ defined on $0 \leq \zeta \leq z \leq 1$, and introducing the new input
\begin{align}
\begin{split}
        \tilde u(t) &\coloneq u(t) + (Q_1 - K(1,1)) x(1,t) \\ & \qquad - \int_{0}^{1} \partial_z K(1,\zeta) x(\zeta,t) \dint{\zeta}, \label{eq:bs_trafo_1_input}
\end{split}
\end{align}
system \eqref{eq:sys_orig} is transformed into the form
\begin{subequations}
        \label{eq:sys_decoup}
        \begin{align}
                \tilde x(z,t) &= \Lambda(z) \partial_z^2 \tilde x(z,t) + A_0(z) \tilde x(0,t) \label{eq:sys_decoup_PDE} \\
                \partial_z \tilde x(0,t) &= Q_0 \tilde x(0,t) \label{eq:sys_decoup_BC0} \\
                \partial_z \tilde x(1,t) &= \tilde u(t). \label{eq:sys_decoup_BC1}
        \end{align}
\end{subequations}
Note that, contrary to \eqref{eq:sys_target_static}, the dynamics \eqref{eq:sys_decoup} does not contain a stabilizing term $-\mu \tilde x(z,t)$ as stabilization is not the objective (yet).
The kernel $K(z,\zeta)$ in \eqref{eq:Volterra_dynamic_prelim} satisfies
\begin{subequations}
        \label{eq:sys_decoup_kernel_equations}
	\begin{align}
		& \Lambda(z) \partial_z^2 K(z,\zeta) - \dd_\zeta^2(K(z,\zeta)\Lambda(\zeta)) = K(z,\zeta) A(\zeta) \label{eq:sys_decoup_ke1} \\
		\begin{split}
                        & \Lambda(z)\dd_z K(z,z) + \Lambda(z)\partial_z K(z,z) \\ & \qquad \quad + \partial_\zeta K(z,z) \Lambda(z) + K(z,z)\Lambda'(z) = -A(z)  \label{eq:sys_decoup_ke2}
                \end{split} \\
		& K(z,z)\Lambda(z) - \Lambda(z)K(z,z) = 0 \label{eq:sys_decoup_ke3} \\
		& \partial_\zeta K(z,0)\Lambda(0) + K(z,0)(\Lambda'(0)-\Lambda(0)Q_0) = -A_0(z) \label{eq:sys_decoup_ke4} \\
		& K(0,0) = 0 \label{eq:sys_decoup_ke5}
	\end{align}
\end{subequations}
on $0 \leq \zeta \leq z \leq 1$.
As in Section~\ref{subsec:static_controller}, a strictly lower triangular matrix $A_0(z) \in \Rset^{n \times n}$, whose nonzero elements are defined by \eqref{eq:sys_decoup_ke4}, is necessary to ensure well-posedness of the kernel equations \eqref{eq:sys_decoup_kernel_equations}.
Existence of a unique piecewise $C^2$-solution is guaranteed by \cite{deutscherBacksteppingControlCoupled2018}.
Note that, in contrast to \cite{deutscherBacksteppingControlCoupled2018}, no restrictions need to be imposed on the boundary coupling matrix $Q_0$ due to the choice of \eqref{eq:sys_decoup_ke5}, allowing for a larger class of systems to be controlled.

\subsection{Homogenization}
\label{subsec:homogenization}
\begin{figure}
        \begin{center}
                        \begin{tikzpicture}[scale=.75]
                        
                                \draw[{Bar[]}-{Bar[]}, thick] (0,0.5) -- (10,0.5) 
                                        node[at start, below=2pt] {$0$}  
                                        node[at end, below=2pt] {$1$};   

                                \coordinate (A1) at (10,2.5);
                                \coordinate (A2) at ( 6,2.5);
                                
                                \coordinate (A) at (A1);
                                \coordinate (A) at (A2);

                                \coordinate (B) at (10,1);

                                \draw [thick, blue, line width=2pt] (0,2.5) -- (A2) node[midway, above] {$\bar x_1(\zbar_1,t)$};
                                \draw [thick, blue, line width=2pt] (0,1) -- (B)  node[midway, above] {$\bar x_2(\zbar_2,t)$};

                                \draw [<-, thick, blue] (A2)+(0,1pt) -- +(0,0.75) node[at end, above=-5pt] {$\tilde u_1(t)$};
                                \draw [dashed, thick, gray] (A2) -- +(0,-2) node[at end, below=2pt] {$\sigma_1(1)$};

                                \draw [thick, orange, line width=2pt] (A2) -- (A1) node[midway, above] {$\bar w_1(\zbar_1,t)$};
                                \draw [<-, thick, orange] (A1)+(0,1pt) -- +(0,0.75) node[at end, above=-5pt] {$v_{1}(t)$};
                                \draw[<-, thick, blue] (B)+(0,1pt) -- +(0,0.75) node[at end, above=-5pt] {$\tilde u_2(t) \mathrel{\textcolor{black}{=}} \highlighteq{orange}{v_{2}(t)}$};
                                
                        \end{tikzpicture}%
        \end{center}
        \caption{Sketch of the idea of homogenization and dynamic extensions for a system \eqref{eq:sys_decoup} with $n=2$. The state component $\bar x_1(\zbar_1,t)$ is defined on the scaled domain $[0,\sigma_1(1)]$ and complemented with the controller state $\bar w_1(\zbar_1,t)$ on $[\sigma_1(1),1]$. The second state component $\bar x_2(\zbar_2,t) = \tilde x_2(z,t)$ remains unchanged. At $z=1$, a new input $v(t)$ is introduced.}
        \label{fig:sketch}
\end{figure}
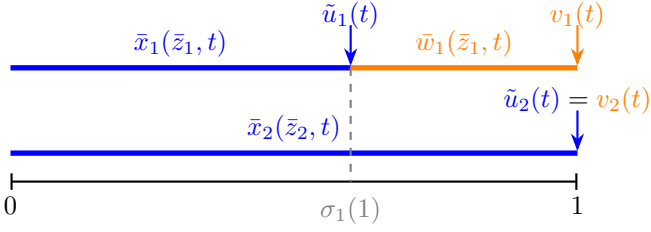

In order to homogenize the diffusion coefficients in \eqref{eq:sys_decoup} in such a way that all diffusion processes share the same diffusion coefficient $\lambda_n$, i.e., the smallest coefficient, we introduce new spatial coordinates $\zbar_i, i=1,\dots,n$ for each individual state component $\tilde x_i(z,t)$.
The idea is to equivalently represent the diffusion of $\tilde x_i(z,t)$ on $[0,1]$ by a new state component
\begin{align}
        \bar x_i(\zbar_i,t) \coloneq \tilde x_i(z,t), \qquad i = 1,\dots,n, \label{eq:homogenization}
\end{align}
defined pointwise on a smaller domain.
The domain of the last state component remains unscaled because the $n$th equation in \eqref{eq:sys_decoup_PDE} already has the diffusion coefficient $\lambda_n$, and therefore $\zbar_n = z$.
This is illustrated for an example with $n=2$ in Figure~\ref{fig:sketch}.

Define for $i=1,\dots,n$ the strictly positive, monotonically increasing function
\begin{align}
        \phi_i(z) \coloneq \int_0^z \frac{1}{\sqrt{\lambda_i(\zeta)}} \dint{\zeta}
\end{align}
with its inverse $\psi_i$ satisfying $\psi_i(\phi_i(z)) = z$.
Then the spatial coordinate is transformed according to
\begin{align}
        \zbar_i = \sigma_i(z) \coloneq \psi_n(\phi_i(z)), \label{eq:zbar}
\end{align}
where $\tau_i$ denotes the inverse transformation of $\sigma_i$ such that $\tau_i(\sigma_i(z)) = z$.
Note that $\sigma_i(0) = 0$ and $0 < \sigma_i(1) \leq 1$ for all $i = 1, \dots, n$ with only $\sigma_n(1) = 1$.


Each state component $\bar x_i(\zbar_i,t)$ is now defined on an interval $[0,\sigma_i(1)]$ of different length.
This way, every diffusion process now has the same diffusion coefficient $\lambda_n$.
Inserting \eqref{eq:homogenization} and its derivatives into \eqref{eq:sys_decoup}, with $\sigma_i'(z) = \sqrt{\tfrac{\lambda_n(\sigma_i(z))}{\lambda_i(z)}}$ in view of \eqref{eq:zbar}, yields the system
\begin{subequations}
        \label{eq:sys_homo}
	\begin{align}
                \begin{split}
		\partial_t \bar x_i(\bar z_i,t) &= \lambda_n(\bar z_i) \partial_{\bar z_i}^2 \bar x_i(\bar z_i,t) + d_i(\bar z_i) \partial_{\bar z_i} \bar x_i(\bar z_i,t) \\ & \qquad + \mathrm{e}_i^\top A_0(\tau_i(\bar z_i)) \bar x(0,t), \quad \zbar_i \in (0,\sigma_i(1)) \label{eq:sys_homo_PDE}
                \end{split}\\
		\partial_{\bar z_i} \bar x_i(0,t) &= \mathrm{e}_i^\top \bar Q_0 \bar x(0,t) \label{eq:sys_homo_BC0} \\
		\partial_{\bar z_i} \bar x_i(\sigma_i(1),t) &= \sqrt{\tfrac{\lambda_i(1)}{\lambda_n(\sigma_i(1))}} \tilde u_i(t), \label{eq:sys_homo_BCsigma1}
	\end{align}
\end{subequations}
$i=1,\dots,n$ with homogenized diffusion coefficients, where $\rme_i^\top \bar Q_0 \coloneq \rme_i^\top \sqrt{\tfrac{\lambda_i(0)}{\lambda_n(0)}} Q_0$ \changes{with $\rme_i$ the $i$th unit vector in $\Rset^n$} and $d_i(\zbar_i) \coloneq \lambda_i(\tau_i(\bar z_i))\sigma_i''(\tau_i(\bar z_i))$.

\subsection{Dynamic Extension}
\label{subsec:dyn_ext}
As the first $i = 1,\dots,n-1$ state components $\bar x_i(\zbar_i,t)$ in \eqref{eq:sys_homo} do not span the ``whole'' interval $[0,1]$ anymore, we introduce state components $\bar w_i(\zbar_i,t)$ of a dynamic extension on the interval $[\sigma_i(1),1]$ so that the composite state\footnote{We use $z$ for terms defined on the interval $[0,1]$ and $\zbar_i$ for expressions, where each component is defined on a different interval.}%
\begin{subequations}%
        \label{eq:chi}%
        \begin{align}%
                \chi_i(z,t) &= \begin{cases}
                        \bar x_i(z,t), & z \in [0,\sigma_i(1)] \\
                        \bar w_i(z,t), & z \in (\sigma_i(1),1]
                \end{cases} \\
                \chi_n(z,t) &= \bar x_n(z,t), \qquad z \in [0,1] \label{eq:chi_n}
        \end{align}
\end{subequations}
with $\chi(z,t) \in \Rset^n$ spans the whole domain $[0,1]$ again (see Figure~\ref{fig:sketch}).
Since the last state component $\bar x_n(\zbar_n,t)$ remains unchanged in view of $\sigma_n(z) = z$, the controller state $\bar w_i(\zbar_i,t)$, $i=1,\dots,n-1$ only has $n-1$ components.

For the composite system with state \eqref{eq:chi} to have ho\-mo\-ge\-nized diffusion coefficients $\lambda_n$ on the domain $[0,1]$, we choose the special dynamics
\begin{align}
        \partial_t \bar w_i(\zbar_i,t) &= \lambda_n(\zbar_i) \partial_{\zbar_i}^2 \bar w_i(\zbar_i,t), \qquad \zbar_i \in (\sigma_i(1),1) \label{eq:dyn_ext_PDE}
\end{align}
for the controller state, with initial condition $\bar w_i(\zbar_i,0) = \bar w_{i,0}(\zbar_i)$ for $i=1,\dots,n-1$.
\begin{remark}
        \label{rmk:dyn_ext}
        The choice of the controller dynamics is a design parameter, e.g., \eqref{eq:dyn_ext_PDE} could at least include additional terms depending on $\bar w_i(\zbar_i,t)$ and $\bar x_i(0,t)$.
\end{remark}
To connect the homogenized system \eqref{eq:sys_homo} and the controller PDEs in \eqref{eq:dyn_ext_PDE}, we require
\begin{subequations}
\begin{align}
        \bar w_i(\sigma_i(1),t) &\overset{!}{=} \bar x_i(\sigma_i(1),t) \label{eq:dyn_ext_BCsigma1} \\
        \partial_{\zbar_i} \bar w_i(\sigma_i(1),t) &\overset{!}{=} \partial_{\zbar_i} \bar x_i(\sigma_i(1),t) \label{eq:dyn_ext_BCsigma1_deriv}
\end{align}
\end{subequations}
for the BCs connecting both dynamics to ensure well-posedness of the resulting system.
While \eqref{eq:dyn_ext_BCsigma1} is a BC of \eqref{eq:dyn_ext_PDE}, \eqref{eq:dyn_ext_BCsigma1_deriv} together with \eqref{eq:sys_homo_BCsigma1} defines the input components
\begin{align}
        \tilde u_i(t) &= \sqrt{\tfrac{\lambda_n(\sigma_i(1))}{\lambda_i(1)}} \partial_{\zbar_i} \bar w_i(\sigma_i(1),t) \label{eq:utilde_BC}
\end{align}
for $i = 1,\dots,n-1$ in terms of the controller state.
At the boundary $\zbar_i=1$, we introduce the new input $v(t)$ (see Figure~\ref{fig:sketch}) such that
\begin{subequations}
        \label{eq:v_BC}
        \begin{align}
                \partial_{\zbar_i} \bar w_i(1,t) &= v_i(t), \qquad i = 1,\dots,n-1 \label{eq:v_BC1} \\
                \tilde u_n(t) &= v_n(t). \label{eq:v_BC2}
        \end{align}
\end{subequations}
In summary, the controller dynamics consists of the PDEs in \eqref{eq:dyn_ext_PDE} with the BCs \eqref{eq:dyn_ext_BCsigma1} and \eqref{eq:v_BC1}.

The dynamically extended system
\begin{subequations}
        \label{eq:sys_ext}
        \begin{align}
                \begin{split}
                        \partial_t \chi(z,t) &= \lambda_n(z) \partial_z^2 \chi(z,t) + \bar D(z) \partial_z \chi(z,t) \\ & \qquad + \bar A_0(z) \chi(0,t)
                \end{split} \label{eq:sys_ext_PDE} \\
                \partial_z \chi(0,t) &= \bar Q_0 \chi(0,t) \label{eq:sys_ext_BC0} \\
                \partial_z \chi(1,t) &= v(t) \label{eq:sys_ext_BC1}
        \end{align}
\end{subequations}
with the state $\chi(z,t)$ defined in \eqref{eq:chi} directly follows from the controller dynamics and the homogenized system \eqref{eq:sys_homo} in view of \eqref{eq:utilde_BC} and \eqref{eq:v_BC2}, with matrices $\bar D = \diag\left(\bar d_1, \dots, \bar d_n\right)$ and $\bar A_0(z)$ defined piecewise by
\begin{subequations}
        \label{eq:sys_ext_definitions}
        \begin{align}
                \bar d_i(z) &= \begin{cases}
                        d_i(z), & z \in [0,\sigma_i(1)] \\
                        0, & z \in (\sigma_i(1),1]
                \end{cases}, \,\, i = 1,\dots,n-1 \\
                \bar d_n(z) &= d_n(z) = 0 \\
                \rme_i^\top \bar A_0(z) &= \begin{cases}
                        \rme_i^\top A_0(\tau_i(z)), & z \in [0,\sigma_i(1)] \\
                        0, & z \in (\sigma_i(1),1]
                \end{cases}, \,\, i = 1,\dots,n.
        \end{align}
\end{subequations}

\subsection{Backstepping Controller}
\label{subsec:backstepping_dynamic}
The advection term $\bar D(z) \partial_z \chi(z,t)$ in \eqref{eq:sys_ext_PDE} is a direct consequence of the scaling of the spatial domains.
To remove it, we use the boundedly-invertible Hopf-Cole-type state transformation (see, e.g., \cite{smyshlyaevClosedFormBoundaryState2004}, \cite{deutscherBacksteppingControlCoupled2018})
\begin{align}
        \begin{split}
                \tilde \chi(z,t) &= \underset{i=1,\dots,n}{\diag} \left(\exp \left(\frac{1}{2} \int_0^z \frac{\bar d_i(\zeta)}{\lambda_n(\zeta)} \dint{\zeta}\right)\right) \chi(z,t) \\ &\eqcolon \Phi(z) \chi(z,t). \label{eq:hopf_cole}
        \end{split}
\end{align}
Applying \eqref{eq:hopf_cole} to \eqref{eq:sys_ext} results in
\begin{subequations}
        \label{eq:sys_hc}
        \begin{align}
                \begin{split}
                        \partial_t \tilde \chi(z,t) &= \lambda_n(z) \partial_z^2 \tilde \chi(z,t) + \bar A(z) \tilde \chi(z,t) \\ & \qquad + \Phi(z) \bar A_0(z) \tilde \chi(0,t)
                \end{split} \label{eq:sys_hc_PDE} \\
                \partial_z \tilde \chi(0,t) &= \left(\bar Q_0 + \Phi'(0)\right) \tilde \chi(0,t) \label{eq:sys_hc_BC0} \\
                \partial_z \tilde \chi(1,t) &= \Phi(1) v(t) \label{eq:sys_hc_BC1}
        \end{align}
\end{subequations}
with $\bar A(z) \coloneq - \lambda_n(z) \Phi''(z) \Phi^{-1}(z)$.
\changes{Note that with the specific choice \eqref{eq:dyn_ext_PDE} for the controller dynamics, it follows from \eqref{eq:sys_ext_definitions} that $\bar D(1) = 0$ and, thus, $\Phi'(1) = 0$.}

Using the (inverse) Volterra integral transformation\footnote{The kernel equations for the inverse kernel take on a simpler form to solve.}
\begin{align}
        \tilde \chi(z,t) &= \bar \chi(z,t) + \int_0^z L(z,\zeta) \bar \chi(\zeta,t) \dint{\zeta} \label{eq:Volterra_dynamic}
\end{align}
with the kernel $L(z,\zeta) \in \Rset^{n \times n}$ defined on $0 \leq \zeta \leq z \leq 1$, system \eqref{eq:sys_hc} is mapped into the desired target system
\begin{subequations}
        \label{eq:sys_target_dynamic}
        \begin{align}
                \partial_t \bar \chi(z,t) &= \lambda_n(z) \partial_z^2 \bar \chi(z,t) + B(z) \bar \chi(z,t) \label{eq:target_dynamic_PDE} \\
                \partial_z \bar \chi(0,t) &= B_0 \bar \chi(0,t) \label{eq:target_dynamic_BC0} \\
                \partial_z \bar \chi(1,t) &= \bar v(t) \label{eq:target_dynamic_BC1}
        \end{align}
\end{subequations}
by applying the control law
\begin{align}
        \begin{split}
                v(t) &= \Phi^{-1}(\changes{1}) \bigg[ \bar v(t) + L(1,1) \bar \chi(1,t) \\ & \qquad \qquad + \int_0^1 \partial_z L(1,\zeta) \bar \chi(\zeta,t) \dint{\zeta} \bigg] \label{eq:control_law_v}
        \end{split}
\end{align}
in \eqref{eq:sys_hc_BC1}.
The matrices $B \in C^1([0,1])^{n \times n}$ and $B_0 \in \Rset^{n \times n}$ as well as the input $\bar v(t) \in \Rset^n$ in \eqref{eq:sys_target_dynamic} constitute design parameters.
Standard calculations reveal that the kernel $L(z,\zeta)$ has to satisfy the kernel equations
\begin{subequations}
        \label{eq:kernel_equations_dynamic}
        \begin{align}
                \begin{split}
                        &\lambda_n(z) \partial_z^2 L(z,\zeta) - \dd_\zeta^2 \left(L(z,\zeta) \lambda_n(\zeta)\right)  \\ & \qquad \qquad = L(z,\zeta) B(\zeta) - \bar A(z) L(z,\zeta)
                \end{split} \\
                &2 \lambda_n(z) \dd_z L(z,z) + \lambda_n'(z) L(z,z) = B(z) - \bar A(z) \\
                \begin{split}
                        &\lambda_n(0) (\partial_\zeta L(z,0) - L(z,0) B_0) + \lambda_n'(0) L(z,0) \\ & \qquad \qquad = \Phi(z) \bar A_0(z)
                \end{split}  \label{eq:kernel_dynamic_A0}\\
                &L(0,0) = \bar Q_0 + \Phi'(0) - B_0 \label{eq:kernel_dynamic_00}
        \end{align}
\end{subequations}
on $0 \leq \zeta \leq z \leq 1$.
Thanks to the homogenized diffusion coefficients, existence of a unique piecewise $C^2$-solution $L(z,\zeta)$ can be shown similar to the equi-diffusivity case in \cite{baccoliBoundaryControlCoupled2015}, where constant coefficients as well as $B_0 = 0$ and $L(0,0) = 0$ are considered.
The solution $L(z,\zeta)$ is found via transformation into canonical coordinates and a successive approximation of the resulting integral equations like in \cite{baccoliBoundaryControlCoupled2015}.
For diagonal matrices $B(z)$ and $B_0$, the kernel equations \eqref{eq:kernel_equations_dynamic} are $n \times n$ independent scalar kernel equations as in \cite{krsticBoundaryControlPDEs2008}.
Note that, contrary to the kernel equations \eqref{eq:sys_decoup_kernel_equations} with mutually different diffusion coefficients, no additional terms are required to ensure well-posedness.
This allows for more freedom in choosing the desired target system \eqref{eq:sys_target_dynamic}.
Furthermore, there is no restriction on the form of $Q_0$.
        
For the target system \eqref{eq:sys_target_dynamic} to be exponentially stable in the $L_2$ norm, one possible choice for the design parameters is $\bar v(t) = 0$, $B_0 = 0$ and $B(z) = - \mu I$ with $\mu > 0$ and identity matrix $I$.
However, stability of the target system can also be achieved with non-diagonal $B(z)$.
For a system with constant diffusion coefficients, this is shown in \cite{baccoliBoundaryControlCoupled2014} for a constant matrix $B$ with positive definite symmetric part.

\begin{theorem}[Closed-loop stability]
        Let $B(z)$, $B_0$ and $\bar v(t)$ be chosen such that the target system \eqref{eq:sys_target_dynamic} is well-posed and exponentially stable in the (weighted) $L_2$ norm of $\bar \chi(\cdot,t)$.
        Then the $L_2$ norm related to the state $x(\cdot,t)$ of system \eqref{eq:sys_orig} converges exponentially to zero if the control input $u(t)$ satisfies \eqref{eq:control_law_v}, \eqref{eq:utilde_BC}, \eqref{eq:v_BC2}, \eqref{eq:bs_trafo_1_input} with the controller dynamics \eqref{eq:dyn_ext_PDE}, \eqref{eq:dyn_ext_BCsigma1}, \eqref{eq:v_BC1}, and $L_2$ initial conditions for system and controller state are compatible with the boundary conditions of the closed loop.
\end{theorem}
\begin{proof}
        As $B(z)$, $B_0$ and $\bar v(t)$ are chosen such that the target system \eqref{eq:sys_target_dynamic} is well-posed so that compatible initial conditions ensure existence of a classical solution, the exponential stability in the $L_2$ norm\footnote{For simplicity, the notation $\lVert \cdot \rVert$ is used for the (weighted) $L_2$ norm of all states, even though the definition of the individual norm depends on the corresponding state space.} yields $\lim_{t \rightarrow \infty} \lVert \bar \chi(\cdot,t) \rVert = 0$.
        From the bounded in\-ver\-ti\-bi\-li\-ty of the transformations \eqref{eq:Volterra_dynamic} and \eqref{eq:hopf_cole}, it can be inferred that $\lim_{t \rightarrow \infty} \lVert \chi(\cdot,t) \rVert = 0$.
        Then \eqref{eq:chi} implies that $\lim_{t \rightarrow \infty} \lVert \bar x_i(\cdot,t) \rVert = 0, \, i=1,\dots,n$ and $\lim_{t \rightarrow \infty} \lVert \bar w_i(\cdot,t) \rVert = 0, \, i=1,\dots,n-1$.
        Consequently, $\lim_{t \rightarrow \infty} \lVert x(\cdot,t) \rVert = 0$ due to \eqref{eq:homogenization} and the bounded invertibility of \eqref{eq:Volterra_dynamic_prelim}.
\end{proof}

Note that diagonal matrices $B(z)$ and $B_0$ result in a system of decoupled diffusion-reaction equations.
This results in a decoupled input-output transfer behavior of \eqref{eq:sys_target_dynamic} from the (new) input $\bar v(t)$ to the output $y(t) = \bar \chi(0,t) = x(0,t)$ (cf.\ \eqref{eq:Volterra_dynamic_prelim}, \eqref{eq:homogenization}, \eqref{eq:hopf_cole}, \eqref{eq:Volterra_dynamic}), which is generally not possible with a static state feedback (see Section~\ref{subsec:static_controller}) because the diffusion processes in \eqref{eq:sys_target_static} are coupled via the matrix $A_0(z)$.
In contrast to the coupled multivariable system \eqref{eq:sys_target_static} in the static case, the target system \eqref{eq:sys_target_dynamic} consists of multiple independent scalar equations.

\addtolength{\textheight}{-3.2cm}   

\section{NUMERICAL RESULTS}
\label{sec:simulation}
A numerical example demonstrates that the dynamic state feedback controller is indeed able to stabilize systems with full boundary coupling matrices $Q_0$ and that it can achieve a decoupled input-output behavior from the input $\bar v(t)$ to the output $x(0,t)$.
We consider system \eqref{eq:sys_orig} with the parameters $\Lambda(z) = \diag(3,2,1)$, $A(z) = \rme^z \cdot 1_{3 \times 3}$, $Q_1 = 0.1 I$, where all elements of $1_{3 \times 3}$ are one.
The matrix $Q_0$ is specified later.

For the simulation, we use \textsc{Matlab}.
Derivatives of the system dynamics \eqref{eq:sys_orig} are approximated using an explicit finite difference scheme with Euler forward differences w.r.t. time and second-order central differences in space.
We choose a time step $\Delta t = \frac{25}{6} \cdot 10^{-4}$ and a CFL number $\lambda_i \frac{\Delta t}{\Delta z_i^2} = \frac{1}{6}$, resulting in different spatial steps\changes{\footnote{\changes{Note that $\Delta t$ and the CFL number are chosen to obtain $\Delta z_n = 0.05$.}}} $\Delta z_i$ for each of the $n=3$ state components of system \eqref{eq:sys_orig}.
The kernel equations \eqref{eq:sys_decoup_kernel_equations} are solved using the \textsc{Matlab} toolbox \cite{backsteppingToolboxParabolisch} and those in \eqref{eq:kernel_equations_dynamic} are solved via transformation into integral equations and a successive approximation scheme like in \cite{baccoliBoundaryControlCoupled2015}.

The control input $u(t)$ is calculated from \eqref{eq:bs_trafo_1_input} with \eqref{eq:utilde_BC} and \eqref{eq:v_BC2} and the controller dynamics \eqref{eq:dyn_ext_PDE}, \eqref{eq:dyn_ext_BCsigma1} and \eqref{eq:v_BC1}, together with the control law \eqref{eq:control_law_v}.
The controller dynamics is approximated with the same scheme as the system dynamics.
Integrals, e.g., in the control law \eqref{eq:control_law_v}, are computed with the trapezoidal rule.

Two different scenarios are analyzed, one which considers stabilization and another one to confirm the input-output decoupling potential.
In both cases, the controller state is initialized as a constant continuation of the homogenized system's boundary values, i.e., $\bar w_i(\zbar_i,0) = \bar x_i(\sigma_i(1),0)$.
Note that due to the (positive) entries of $A(z)$ and the signs of the entries of $Q_0$ and $Q_1$, the uncontrolled system with $u(t) = 0$ is inherently unstable.

\subsection{Stabilization}
\label{subsec:sim_stabi}
In order to stabilize system \eqref{eq:sys_orig} with $Q_0 = -0.1 \cdot 1_{3 \times 3}$, the parameters of the target system \eqref{eq:sys_target_dynamic} are chosen as $\bar v(t) = 0$, $B_0 = 0$ and $B(z) = -I$.
We set initial conditions $x_i(z,0) = i \cdot \rme^{-(z-0.2)^2} + c_{1,i} z^2 + c_{2,i} z$ with $ i=1,2,3$, where $c_{1,i}$ and $c_{2,i}$ are chosen such that the BCs \eqref{eq:sys_orig_BC0} and \eqref{eq:sys_orig_BC1} are satisfied, which ensures a classical solution.
Figure~\ref{fig:l2norm_Q0full} shows the (weighted) $L_2$ norm of the state $x(z,t)$.
It can be seen that the norm decays and converges to zero, which confirms that the unstable system \eqref{eq:sys_orig} is indeed stabilized.

\begin{figure}
        \centering
        \vspace{.25cm}
        \includegraphics[width=\linewidth]{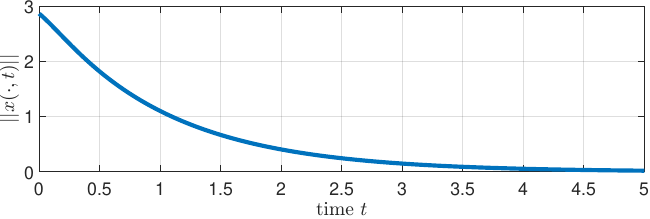}
        \caption{(Weighted) $L_2$ norm $\lVert x(\cdot,t) \rVert$ of the state in the closed loop with a dynamic state feedback where $\bar v(z) = 0$, $B_0 = 0$, $B(z) = -I$.}
        \label{fig:l2norm_Q0full}
\end{figure}

\subsection{Decoupling}
\label{subsec:sim_decoupling}
To demonstrate the decoupling potential of the dynamic controller, we compare it to the static state feedback controller from Section~\ref{subsec:static_controller}.
To this end, we set $Q_0 = -0.1 I$ because the static controller in \cite{deutscherBacksteppingControlCoupled2018} as described in \mbox{Section~\ref{subsec:static_controller}} requires $Q_0$ to be diagonal.
For the case of static state feedback, we set $\mu=1$ in \eqref{eq:sys_target_static}, analog to $B(z) = -I$ and $B_0 = 0$ for the dynamic state feedback.
We assume for simplicity that $x(z,0) = 0 , \, \forall z \in [0,1]$, and use a unit step for $\bar v_1(t)$ while keeping $\bar v_2(t) = \bar v_3(t) = 0, \, \forall t\geq0$.
This is depicted in Figure~\ref{fig:traj_step} on the left.
On the right, we can see the output $x(0,t)$ for both the static and the dynamic state feedback.
When using a dynamic feedback controller, the input-output behavior is clearly decoupled.
Only the first output component $x_1(0,t)$ reacts to the step in the first input component $\bar v_1(t)$, while the other output components remain at zero.
This is not the case for the static state feedback controller due to the coupling between the state components introduced by the stricly lower triangular matrix $A_0(z)$ in \eqref{eq:sys_target_static_PDE}.
Note that the different trajectories of $x_1(0,t)$ are because of the different location of $\bar v_1(t)$ in the static and the dynamic design, respectively (see \eqref{eq:sys_target_static} and \eqref{eq:sys_target_dynamic}).
\begin{figure*}
        \begin{center}
                \vspace{.25cm}
                \includegraphics[width=\textwidth]{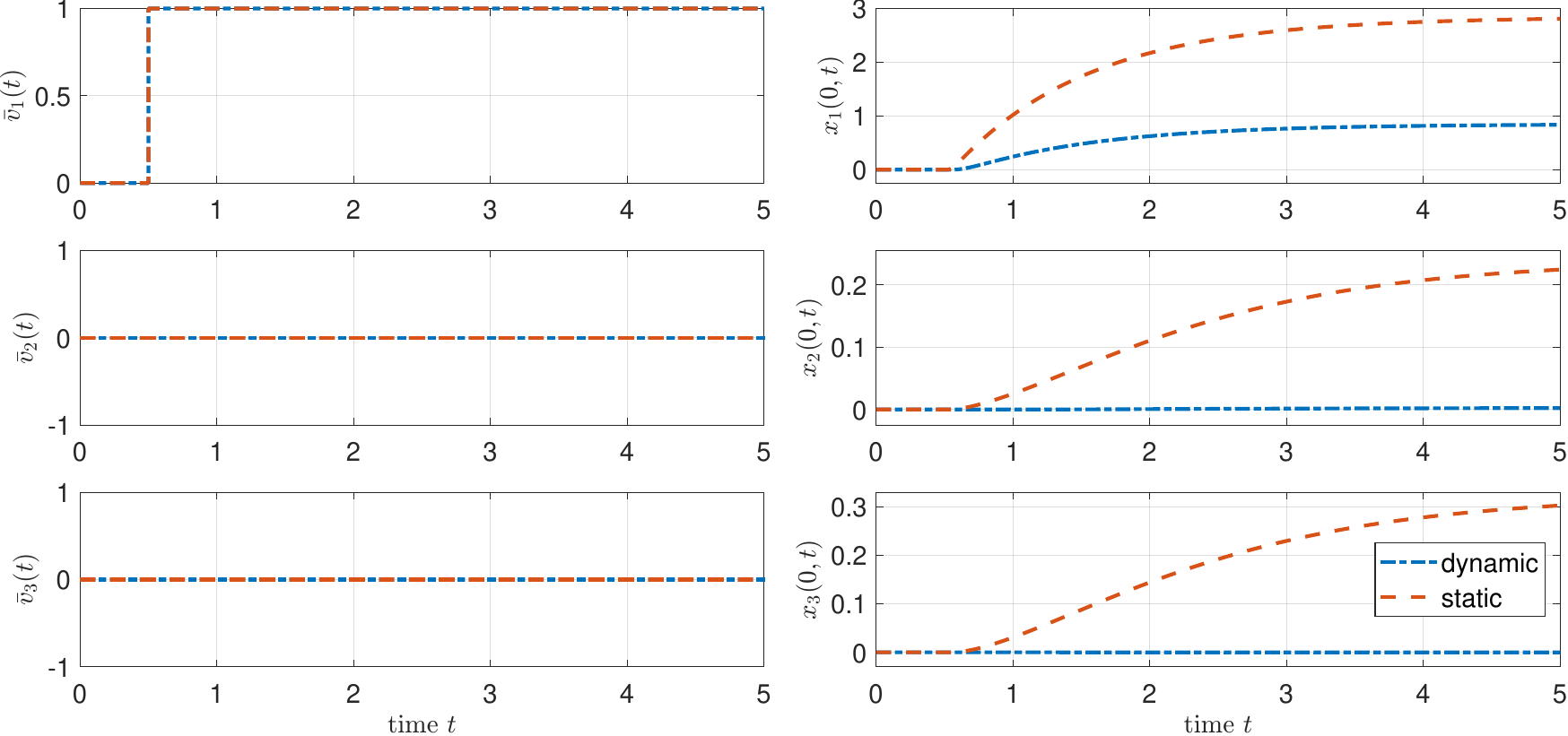}
                \caption{Unit step in $\bar v_1(t)$ and $\bar v_2(t) = \bar v_3(t) = 0$ (left) and step response of the output $x(0,t)$ of system \eqref{eq:sys_orig} (right). For the dynamic state feedback, we set $B(z) = -I$ and $B_0=0$, and for the static state feedback, we set $\mu=1$.}
                \label{fig:traj_step}
        \end{center}
\end{figure*}

\section{CONCLUDING REMARKS}
\label{sec:conclusion_outlook}
In this paper, we present the design of a dynamic controller for coupled diffusion-reaction systems.
We introduce a dynamic extension and a backstepping-based static state feedback for the dynamically extended system.
As there are no restrictions on the form of $Q_0$, this allows stabilizing a larger class of boundary-coupled parabolic systems.
Further, this design gives more freedom in choosing the target system, thus enabling to assign more general dynamics to the closed-loop system.
This can be used, for instance, for decoupling of the transfer behavior from the input $\bar v(t)$ to the output $x(0,t)$.
Future work will focus on analyzing different choices of the controller dynamics \eqref{eq:dyn_ext_PDE} and of the target system \eqref{eq:sys_target_dynamic}, and what other possibilities and applications this offers.

\section{ACKNOWLEDGMENTS}
This research was funded in part by the Deutsche Forschungsgemeinschaft (DFG, German Research Foundation) under project no. 517291864 and the Austrian Science Fund (FWF) [I 6519-N].



\printbibliography

%

\end{document}

%% file: mycommands.tex
\newcommand{\rme}{\mathrm{e}}				
\newcommand{\Rset}{\mathbb R}				

\newcommand\dint[1]{\,\mathrm{d}{#1}}
\newcommand\dd{\mathrm{d}}

\newcommand\diag{\text{diag}}

\newcommand{\highlighteq}[2]{%
  \textcolor{#1}{#2}}

\newtcolorbox{ovgubox}[2][]{
  colframe=ovguPurpleDark,
  colback=ovguPurpleLight,          
  colbacktitle=ovguPurpleDark,
  coltitle=white,
  rounded corners,
  boxrule=0.4pt,
  left=0pt,
  right=0pt,
  top=0pt,
  bottom=0pt,
  before skip=10pt,
  after skip=10pt,
  title={#2},
  fonttitle=\bfseries, 	  
  arc=5pt,
  valign=center,
  #1
}